\documentclass[10pt,a4paper]{article}
\usepackage[utf8]{inputenc}
\usepackage[english]{babel}
\usepackage{amsmath}
\usepackage{amsfonts}
\usepackage{amssymb}
\usepackage{amsthm}
\usepackage{appendix}
\usepackage{cite}
\usepackage{graphicx}
\usepackage[left=2cm,right=2cm,top=2cm,bottom=2cm]{geometry}
\newtheorem{te}{Theorem}

\newtheorem{lem}{Lemma}

\newcommand{\pa}{\partial}
\newtheorem{prop}{Proposition}

\newtheorem{de}{Definition}

\newcommand{\bbf}[1]{\mbox{\boldmath{$#1$}}}
\providecommand{\keywords}[1]
{\small	\textbf{\textit{Keywords:}} #1}
\author{A. M. Escobar-Ruiz and R. Azuaje\\
Departamento de F\'{i}sica, Universidad Aut\'onoma Metropolitana Unidad Iztapalapa\\ San Rafael Atlixco 186, 09340 Cd. Mx., M\'exico}
\title{On particular integrability in classical mechanics}
\begin{document}
\maketitle

\begin{abstract}
In this study the notion of particular integrability in Classical Mechanics, introduced in [J. Phys. A: Math. Theor. 46 025203, 2013], is revisited within the formalism of symplectic geometry. A particular integral $\cal I$ is a function not necessarily conserved in the whole phase space $T^*Q$ but when restricted to a certain invariant subspace ${\cal W}\subseteq   T^*Q$ it becomes a Liouville first integral. For natural Hamiltonian systems, it is demonstrated that such a function $\cal I$ allows us to construct a lower dimensional Hamiltonian in $\cal W$. This symmetry reduction is intimately related with a phenomenon beyond separation of variables and it is based on an adaptive application of the classical results due to Lie and Liouville on integrability. Three physically relevant systems are used to illustrate the underlying  key aspects of the symplectic theory approach to particular integrability: (I) the integrable central-force problem, (II) the chaotic two-body Coulomb system in a constant magnetic field as well as (III) the $N$-body system.

\keywords{Hamiltonian systems, particular integrability, symplectic geometry, Lie's theorem, integrals of motion.}

\end{abstract}

\section{Introduction}
\label{sec1}

Since its inception the theory of Hamiltonian systems plays a fundamental role in Theoretical Physics. In particular, it is standard material in any advanced textbook on Classical Mechanics \cite{Landau} where a fruitful relationship to symplectic geometry occurs \cite{Arnold89,MacKay}. Moreover, the Hamiltonian formalism emerges as a natural link between classical and quantum mechanics. For a Hamiltonian system with $n$ degrees of freedom, the existence of $n$ integrals of motion that are almost everywhere independent and that are in involution is required to make it integrable. The property of complete integrability, in the Liouville sense, implies that it is locally equivalent to an action-angle system. Therefore, it is intimately related with the mark of solvability. 
Indeed, integrability provides a regular foliation of the phase space which in turn would allow us to find the solutions, of the corresponding Hamilton's equations of motion, by quadratures in a systematic way. Along these lines, a superintegrable Hamiltonian system admits $k$ extra independent integrals being $k=n-1$ the maximum possible number in which case all the bounded trajectories are periodic \cite{Nek72} and can be obtained, in principle, by pure algebraic means. The notions of integrability as well as that of superintegrability can be defined in quantum and classical mechanics, and there exist several exhaustive review 
articles in literature about the history and current status of this topic \cite{bolsinov2004, BBT2003, HW, MillerPostWinternitz:2013}.

For instance, in the case of a spherically symmetric integrable classical Hamiltonian in $\mathbb{R}^n$ ({\it i.e.}, a potential of the form  $V=V(r)$, being $r$ the radial distance), there exist 
two superintegrable systems only: the celebrated Kepler system (Newtonian gravity) and the simple harmonic oscillator (Hooke's law). In fact, these two central potentials are the ones 
which appear in the Bertrand's theorem \cite{Bertrand1873,GPS2002}. The superintegrability of the Kepler system is due to the 
existence of the conserved Laplace-Runge-Lenz vector \cite{GPS2002, Laplace, Lenz, Runge} whereas in the case of the harmonic oscillator it is a direct consequence of the also conserved quadrupole Jauch-Fradkin tensor \cite{Fradkin}.

Recently, in \cite{AM2021, AM2022} concrete examples where only for specific initial conditions certain quantities become integrals of motion were observed in the study of the physically relevant $3-$ and $4-$body classical system. This pattern realizes the notion of a \textit{particular integral} introduced in \cite{ParIntTurbiner2013} (see also references \cite{TurbinerVieyra2020, TurbinerVieyra2021}). Eventually, if the interaction potential between the bodies solely depends on their relative distances $r_{ij}$, as in the case of the $N-$body problem in Celestial Mechanics, it is possible to construct a universal reduced Hamiltonian in these variables which describes all the trajectories possessing zero total angular momentum \cite{AM2021, AM2022, Miller2018}. In the corresponding derivation, as a key element certain particular integrals were promoted to new generalized coordinates following the spirit of the Hamilton-Jacobi theory. 

Hence, one important reason to study a particular integral is because it can be exploited to reduce the number of degrees of freedom. Accordingly, the task of solving the corresponding set of dynamical equations becomes simpler. 
While all the motion of complete integrable Hamiltonian systems, assuming that the regular level set of the invariants is compact, takes place in an invariant torus, for the dynamics of a non-integrable systems there may exist only a part of the total dynamics where the requirements for integrability hold. Therefore, the question on a formalism to describe the latter situation naturally arises.

The main goal of the present study is twofold. Firstly, we aim to establish in a rigorous manner, within the formalism of symplectic geometry, the above-mentioned notion of \textit{particular} integral. It is the basic element in the also revisited concept of \textit{particular integrability}. Secondly, for autonomous Hamiltonian systems we plan to describe the associated reduction of the corresponding set of dynamical equations by virtue of the existence of particular integrals. Especially, to indicate the conditions under which the reduced system of equations is governed by a reduced Hamiltonian. The present modest contribution can be seen as an adaptive application of the classical results due to Lie and Liouville on integrability. 
It is worth mentioning that our approach is related with a phenomenon beyond the standard separation of variables and it can be used to characterize integrable sub-structures of non-integrable Hamiltonian systems. In order to clarify the ideas of the present work, the formalism is applied to physically significant examples such as the integrable central-force problem, the chaotic two-body Coulomb system in a magnetic field, and the $N-$body system.

\subsection{Generalities}

As a first step, let us recall the basics on symplectic geometry and the geometric formulation of Hamiltonian Mechanics (for details see References \cite{AMRC2019,LR89,Torres2020,Lee2012,MR99}).
 
Let $M$ be a smooth manifold. A symplectic structure on $M$ is a closed non-degenerate 2-form $\omega$ on $M$. Closed means that $d\omega=0$ and non-degenerate implies that for each 1-form $\alpha$ on $M$ there is one and only one vector field $X$ which obeys $X\lrcorner\omega=\alpha$. A symplectic manifold is a pair $(M,\omega)$ where $\omega$ is a symplectic structure on $M$. By definition, each symplectic manifold is of even dimension.

Let $(M,\omega)$ be a symplectic manifold of dimension $2n$. Around any point $p\in M$ there exist local coordinates $(q^{1},\cdots,q^{n},p_{1},\cdots,p_{n})$, called canonical coordinates or Darboux coordinates, such that
\begin{equation}
\omega\ = \ dq^{i}\wedge dp_{i}\ .
\end{equation}
In this paper, we adopt the Einstein summation (\textit{i.e.}, a summation over repeated indices is assumed). In Classical Mechanics, the (contravariant) variables $\{ q^{n} \}$ are called generalized coordinates whilst the (covariant) quantities $\{ p_{n} \}$ correspond to their canonical momenta, respectively.   

For each function $f(p,q)\in C^{\infty}(M)$ in the phase space, is assigned a vector field $X_{f}$ on $M$, called the Hamiltonian vector field for $f$, according to
\begin{equation}
X_{f}\lrcorner \omega \ = \ df \ .
\end{equation}
In canonical coordinates $(q,p)$, this vector field $X_{f}$ takes the form
\begin{equation}
X_{f}=\frac{\partial f}{\partial p_{i}}\frac{\partial}{\partial q^{i}}-\frac{\partial f}{\partial q^{i}}\frac{\partial}{\partial p_{i}}\ .
\end{equation}
The assignment $f\longmapsto X_{f}$ is linear, namely
\begin{equation}
X_{f+\alpha \,g}\ = \ X_{f}+\alpha\, X_{g}\ ,
\end{equation}
$\forall\, f,g\in C^{\infty}(M)$ and $\forall\, \alpha \in\mathbb{R}$. For given functions $f,g \in C^{\infty}(M)$, the Poisson bracket of $f$ and $g$ is defined by
\begin{equation}
\label{PB}
\lbrace f,g\rbrace\ = \ X_{g}\,f \ = \ \omega(X_{f},X_{g})\ .
\end{equation}
In canonical coordinates, we have
\begin{equation}
\label{CPB}
\lbrace f,g\rbrace \ = \ \frac{\partial f}{\partial q^{i}}\frac{\partial g}{\partial p_{i}}\,-\,\frac{\partial f}{\partial p_{i}}\frac{\partial g}{\partial q^{i}} \ .
\end{equation}

The geometric theory of conservative Hamiltonian systems, when the Hamiltonian $H$ is not an explicit function of time $ t $, is naturally constructed within the mathematical formalism of symplectic geometry. Given $H\in C^{\infty}(M)$ the dynamics of the Hamiltonian system on $(M,\omega)$ (the phase space) is governed by the Hamiltonian vector field $X_{H}$. In this case, we say that $(M,\omega,H)$ is a Hamiltonian system with $n$ degrees of freedom, and the trajectories of the system $\psi(t)=(q^{1}(t),\cdots,q^{n}(t),p_{1}(t),\cdots,p_{n}(t))$ are the integral curves of $X_{H}$. They satisfy the Hamilton's equations of motion
\begin{equation}
\dot{q^{i}} =\frac{\partial H}{\partial p_{i}}, \hspace{1cm}
\dot{p_{i}} =-\frac{\partial H}{\partial q^{i}}\qquad ;\qquad i=1,2,3,\ldots,n \ .
\end{equation}

The evolution (the temporal evolution) of a function $f\in C^{\infty}(M)$ (a physical observable) along the trajectories of the system is given by
\begin{equation}
\dot{f}= \mathcal{L}_{X_{H}}f=X_{H}f=\lbrace f,H\rbrace \ ,
\end{equation}
where $\mathcal{L}_{X_{H}}f$ is the Lie derivative of $f$ with respect to $X_{H}$. A function $f$ is a global constant of motion of the system (or an integral of motion) if it is constant along all the trajectories of the Hamiltonian system, that is, $f$ is a constant of motion if $\mathcal{L}_{X_{H}}f=0$ (or equivalently $\lbrace f,H\rbrace=0$). Clearly, a global integral is conserved for any set of initial conditions.

The existence of such a global integrals of motion can be used to systematically reduce the number of degrees of freedom. Thus, the task of solving the corresponding Hamilton's equation of motion becomes simpler. For instance, the aforementioned symmetry reduction is exploited in the alternative Hamilton-Jacobi theory which is based on selecting the global integrals of motion as the new momenta. Eventually, the Hamiltonian is transformed to a suitable form such that the Hamilton's equations become directly integrable. 

In the next section, we aim to define quantities which are conserved only for specific initial conditions. Hence, they are not global integrals. Nevertheless, in some invariant manifolds of the phase space they also can be used to reduce the number of degrees of freedom. This will be illustrated using physically relevant applications.

\section{Hamiltonian symmetry reduction using \textit{particular integrals}}
\label{sec2}

\subsection{Particular integral \textit{vs} global integral}

In this Section, within the formalism of symplectic geometry, we present the concept of a particular integral \cite{ParIntTurbiner2013}. Let $(M,\omega,H)$ be a Hamiltonian system with $n$ degrees of freedom (${\rm dim}(M)=2n$). 

For a given $f\in C^{\infty}(M)$, let us consider the level set 
\[ 
M_{f} \ = \ \lbrace\, x\in M: \ f(x)\,=\,0\,\rbrace \ ,
\]
if $0$ is a regular value of $f$ then $M_{f}$ is a smooth submanifold of $M$ of dimension $2n-1$ \cite{Lee2012}.
\begin{de}
We say that a function (and observable) $f\in C^{\infty}(M)$ is a particular integral of $(M,\omega,H)$ if 
\[
X_{H}\,f\ =\ a\,f \qquad(\text{or equivalently}\ \lbrace f,H\rbrace=a\,f) \ , 
\]
for some $a\in C^{\infty}(M)$ such that \[ \lim_{x\to x_{0}} a(x) \ = \ a(x_{0})\in\mathbb{R}\ , \quad \forall 
 \,x_{0}\in M_{f} \ .\]
If so, we say that $a\in C^{\infty}(M)$ is a function with real values on $f=0$.
\end{de}
The simplest case $a=0$, corresponds to the definition of a first integral $f$ in the Liouville sense. In this instance, we say that $f$ is a global integral. It is a conserved quantity independently of the values of the initial conditions. Now, let us suppose that $0$ is a regular value of $f$.
\begin{lem}
If $f$ is a particular integral of $(M,\omega,H)$ then $M_{f}$ is closed under the dynamics of $X_{H}$, i.e., if $\gamma:I\longrightarrow M$ is an integral curve of $X_{H}$ such that $\gamma(t_{0})\in M_{f}$ for some $t_{0}\in I$ then $\gamma(t)\in M_{f}$ for every $t\in I$.
\end{lem}
\begin{proof}
Let $\gamma:I\longrightarrow M$ be an integral curve of $X_{H}$ such that $\gamma(t_{0})\in M_{f}$ for some $t_{0}\in I$, so we have $f(\gamma(t_{0}))=0$ and
\begin{equation}
\begin{split}
\frac{d}{dt}f(\gamma(t))&=df_{\gamma(t)}(\dot{\gamma}(t))\\
&=(\dot{\gamma}(t))(f)\\
&=(X_{H}f)(\gamma(t))\\
&=(af)(\gamma(t))\\
&=a(\gamma(t))f(\gamma(t)) \ ,
\end{split}
\end{equation}
hence, $\frac{d}{dt}f(\gamma(t))$ is proportional to $f(\gamma(t))$. Similarly, $\frac{d^{k}}{dt^{k}}f(\gamma(t))$ is proportional to $f(\gamma(t))$ for each $k=1,2,\ldots$, so we have that
\begin{equation}
\frac{d^{k}}{dt^{k}}f(\gamma(t))|_{t=t_{0}}\ = \ 0 \ ,
\end{equation}
for every $k=1,2,\ldots$, and $f(\gamma(t))$ as well as all its derivatives vanish on $t=t_{0}$. Since we assume that $f(\gamma(t))$ is an analytic function then $f(\gamma(t))=0$ $\forall\, t\in I$ \cite{Lima,Rudin,Ross}, i.e., $\gamma(t)\in M_{f} \ \forall\, t\in I$.
\end{proof}

The defining condition for a particular integral, namely $\lbrace f,H\rbrace=af$, has been presented in \cite{BCM2012,KKM2018} as the consistency requirement for looking for solutions of the Hamilton-Jacobi equation through the imposition of the so called \textit{side conditions}. Explicitly, suppose we have the Hamilton-Jacobi equation in certain orthogonal coordinates $(x^{i})$
\begin{equation}
\label{HJeq}
 g^{jj}(x)\bigg(\frac{\partial u}{\partial x^{j}}\bigg)^{2}\ + \ V(x) \ = \ E \ .
\end{equation}
If $S(x,p)$ is a well-defined function in phase space, quadratic in the momenta $(p_{i})$ and diagonalizable in the same coordinates $(x^{i})$, such that $\lbrace S,H\rbrace=a\,S$ for some function $a$, then it allows us to look for solutions $u(x)$ of the Hamilton-Jacobi equation (\ref{HJeq}) subjected to the side condition $S(x,p)=0$ (for further details see \cite{BCM2012}). In particular, it will lead to regular and non-regular separation of variables in (\ref{HJeq}). Our present approach presents two important differences with respect to the treatment described in \cite{BCM2012}. Firstly, the formalism is mainly based on the underlying symplectic geometry of the system rather than on the separability of the Hamilton-Jacobi equation. Secondly, it is more general since we consider Hamiltonian functions not necessarily quadratic in the momentum variables.    

\subsection{Towards the construction of a reduced Hamiltonian}

The existence of a particular integral enable us to find solutions of the Hamilton's equations of motion that are solutions of a reduced dynamical system. Indeed, since $M_{f}$ is closed under the dynamics of $X_{H}$, we can look for integral curves of $X_{H}$ living in $M_{f}$ only. These solutions are not the most general ones. However, they are solutions of a reduced system of differential equations. This is demonstrated in the following Theorem.
\begin{te}
If $f$ is a particular integral of $(M,\omega,H)$ then the solutions of the Hamilton's equations of motion that live in $M_{f}$ can be found by solving a system of $2n-1$ differential equations and quadratures.
\end{te}
\begin{proof}
Let $f$ be a particular integral of $(M,\omega,H)$. If we have a canonical transformation $(q,p)\mapsto (Q,P)$ such that $P_{n}=f$, then the Hamilton's equations of motion in the new coordinates $(Q,P)$ take the form
\begin{equation}
\label{eqK}
\left\lbrace \begin{array}{c}
\dot{Q}^{1}= \frac{\partial K}{\partial P_{1}}\\ 
\vdots \\ 
\dot{Q}^{n-1}= \frac{\partial K}{\partial P_{n-1}}\\ 
\dot{Q}^{n}= \frac{\partial K}{\partial f}\\ 
\dot{P}_{1}= -\frac{\partial K}{\partial Q^{1}}\\ 
\vdots \\ 
\dot{P}_{n-1}= -\frac{\partial K}{\partial Q^{n-1}}\\ 
\dot{P}_{n}=\dot{f}= -\frac{\partial K}{\partial Q^{n}}
\end{array} \right. \ ,
\end{equation}
with $K=H(Q,P)$ the new Hamiltonian function. The corresponding canonical transformation can be described by a generating function of the form $F=F_{2}(q^{1},\cdots,q^{n},P_{1},\cdots,P_{n-1},f)-Q^{i}P_{i}$ \cite{GPS2002,Landau,Calkin}. Since we are looking for solutions of the Hamilton's equations of motion that live in $M_{f}$ then we set up $f=0$. Consequently, at $f=0$ the system (\ref{eqK}) reduces to $P_{n}=f=0$ and the $2n-1$ differential equations:
\begin{equation}
\label{reducedsys}
\left\lbrace \begin{array}{c}
\dot{Q}^{1}= \frac{\partial K}{\partial P_{1}}\big|_{{}_{f=0}}\\ 
\vdots \\ 
\dot{Q}^{n-1}= \frac{\partial K}{\partial P_{n-1}}\big|_{{}_{f=0}}\\
\dot{Q}^{n}= \frac{\partial K}{\partial f}\big|_{{}_{f=0}}\\ 
\dot{P}_{1}= -\frac{\partial K}{\partial Q^{1}}\big|_{{}_{f=0}}\\ 
\vdots \\ 
\dot{P}_{n-1}= -\frac{\partial K}{\partial Q^{n-1}}\big|_{{}_{f=0}}\\ 
\end{array} \right. \ .
\end{equation}
\end{proof}
It is worth remarking that system (\ref{reducedsys}) is not a system of Hamilton's equations of motion since it is a system of an odd number of differential equations. In fact, $(Q^{1},\cdots,Q^{n},P_{1},\cdots,P_{n-1})$ are coordinates on $M_{f}$ which is a smooth manifold of odd dimension $2n-1$. Hence, $M_{f}$ is not symplectic. Here, (\ref{reducedsys}) are the equations of motion for the autonomous dynamical system defined by $X_{H}|_{M_{f}}$, i.e., the system specified by the action of $X_{H}$ on $M_{f}$ which is not Hamiltonian. Some remarks are in order:
\begin{itemize}
    \item Since we are taking $P_{n}=f=0$, thus, $\dot{P}_{n}=-\frac{\partial K}{\partial Q^{n}}\big|_{{}_{f=0}}=0$, the function $K\big|_{M_{f}}$ does not depend on the coordinate $Q^{n}$. Therefore, the time evolution of each coordinate $Q^{1},\cdots,Q^{n-1},P_{1},\cdots,P_{n-1}$ on $M_{f}$ is independent of $Q^{n}$. 
    \item Because $P_{n}=f$ is not a coordinate on $M_{f}$, formally the term $\frac{\partial K}{\partial f}\big|_{{}_{f=0}}$ can not be calculated as a derivative on $M_{f}$. In fact, it is a function defined on $M$ restricted to $M_{f}$. Hence, it follows that the time evolution of $Q^{n}$, $\dot{Q}^{n}= \frac{\partial K}{\partial f}\big|_{{}_{f=0}}$, may depend on $Q^{n}$ as well.
    \item One can solve the original system (\ref{reducedsys}) by solving first the reduced system of $2(n-1)$ equations 
    \begin{equation}
\label{reducedhamsys}
\left\lbrace \begin{array}{c}
\dot{Q}^{1}= \frac{\partial K}{\partial P_{1}}\big|_{{}_{f=0}}\\ 
\vdots \\ 
\dot{Q}^{n-1}= \frac{\partial K}{\partial P_{n-1}}\big|_{{}_{f=0}}\\
\dot{P}_{1}= -\frac{\partial K}{\partial Q^{1}}\big|_{{}_{f=0}}\\ 
\vdots \\ 
\dot{P}_{n-1}= -\frac{\partial K}{\partial Q^{n-1}}\big|_{{}_{f=0}}\\ 
\end{array} \right.  \ ,
\end{equation}
     and afterwards integrating the equation $\dot{Q}^{n}=\frac{\partial K}{\partial f}(Q^{1},\cdots,Q^{n-1},Q^{n},P_{1},\cdots,P_{n-1})$.
\end{itemize}

System (\ref{reducedhamsys}) is a system of Hamilton's equations of motion in the coordinates $(Q^{1},\cdots,Q^{n-1},P_{1},\cdots,P_{n-1})$. Indeed, we have that $(Q^{1},\cdots,Q^{n-1},P_{1},\cdots,P_{n-1})$ are local coordinates for the submanifold $M_{f,Q}=\lbrace x\in M: \ f(x)=0, \ Q^{n}(x)={\rm constant} \rbrace$ and $(M_{f,Q},\omega\big|_{M_{f,Q}})$ is a symplectic manifold of dimension $2(n-1)$. We can see that the expression of $\omega\big|_{M_{f,Q}}$ in the local coordinates $(Q^{1},\cdots,Q^{n-1},P_{1},\cdots,P_{n-1})$ is $dQ^{1}\wedge dP_{1}+\cdots+dQ^{n-1}\wedge dP_{n-1}$. Therefore, (\ref{reducedhamsys}) is the system of Hamilton's equations of motion in the lower-dimensional phase space $(M_{f,Q},\omega\big|_{M_{f,Q}})$ with Hamiltonian function $K\big|_{M_{f}}=K\big|_{M_{f}}(Q^{1},\cdots,Q^{n-1},P_{1},\cdots,P_{n-1})$. The dynamics defined by the Hamiltonian system $(M_{f,Q},\omega\big|_{M_{f,Q}},K\big|_{M_{f}})$ is the projection of the dynamics of $X_{H}\big|_{M_{f}}$ into the submanifold $M_{f,Q}$ in the sense that if $\pi:M_{f}\longrightarrow M_{f,Q}$ is the projection map from $M_{f}$ to the submanifold $M_{f,Q}$ defined in the canonical coordinates $(Q,P)$ by 
\begin{equation}
    \pi(Q^{1},\cdots,Q^{n-1},Q^{n},P_{1},\cdots,P_{n-1})\ = \ (Q^{1},\cdots,Q^{n-1},P_{1},\cdots,P_{n-1})\ ,
\end{equation}
then the trajectories of the system $(M_{f,Q},\omega\big|_{M_{f,Q}},K\big|_{M_{f}})$ are images by $\pi$ of trajectories of the system defined by $X_{H}$ on $M_{f}$, i.e., $\gamma(t)=(\,Q^{1}(t),\cdots,Q^{n-1}(t),P_{1}(t),\cdots,P_{n-1}(t)\,)$ is a trajectory of the system $(M_{f,Q},\omega\big|_{M_{f,Q}},K\big|_{M_{f}})$ if and only if there is a trajectory $\overline{\gamma(t)}=(Q^{1}(t),\cdots,Q^{n-1}(t),Q^{n}(t),P_{1}(t),\cdots,P_{n-1}(t))$ such that $\pi(\overline{\gamma(t)})=\gamma(t)$ for every $t$, see Figure \ref{FMfQ}. 

To lift the dynamics from $M_{f,Q}$ to $M_{f}$ we consider trajectories $\gamma(t)=(Q^{1}(t),\cdots,Q^{n-1}(t),P_{1}(t),\cdots,P_{n-1}(t))$ of the system $(M_{f,Q},\omega\big|_{M_{f,Q}},K\big|_{M_{f}})$ and construct trajectories $\overline{\gamma(t)}=(Q^{1}(t),\cdots,Q^{n-1}(t),Q^{n}(t),P_{1}(t),\cdots,P_{n-1}(t))$ on $M_{f}$ by integrating the equation $\dot{Q}^{n}=\frac{\partial K}{\partial f}(Q^{1},\cdots,Q^{n},P_{1},\cdots,P_{n-1})$. It is worth mentioning that this reduction process is possible due to the fact that the time evolution of each coordinate $Q^{1},\cdots,Q^{n-1},P_{1},\cdots,P_{n-1}$ is independent of $Q^{n}$. Summarizing:
\begin{itemize}
    \item $(Q^{1},\cdots,Q^{n-1},P_{1},\cdots,P_{n-1})$ are local coordinates on $M_{f,Q}$, and (\ref{reducedhamsys}) is the system of Hamilton's equations of motion for the reduced Hamiltonian system with phase space $(M_{f,Q},\omega\big|_{M_{f,Q}})$ and Hamiltonian function $K\big|_{M_{f}}=K\big|_{M_{f}}(Q^{1},\cdots,Q^{n-1},P_{1},\cdots,P_{n-1})$.
    \item The dynamics defined by the Hamiltonian system $(M_{f,Q},\omega\big|_{M_{f,Q}},K\big|_{M_{f}})$ is the projection of the dynamics of $X_{H}\big|_{M_{f}}$ into the submanifold $M_{f,Q}$. In general, it is not contained into the dynamics defined by the vector field $X_{H}$ on $M_{f}$, see Fig. \ref{FMfQ}. Nevertheless, we still have a reduction process since the Hamiltonian system $(M_{f,Q},\omega\big|_{M_{f,Q}},K\big|_{M_{f}})$ is a system with $n-1$ degrees of freedom.
\end{itemize}

\begin{figure} [ht!]
    \centering
    \includegraphics[width=0.9\textwidth]{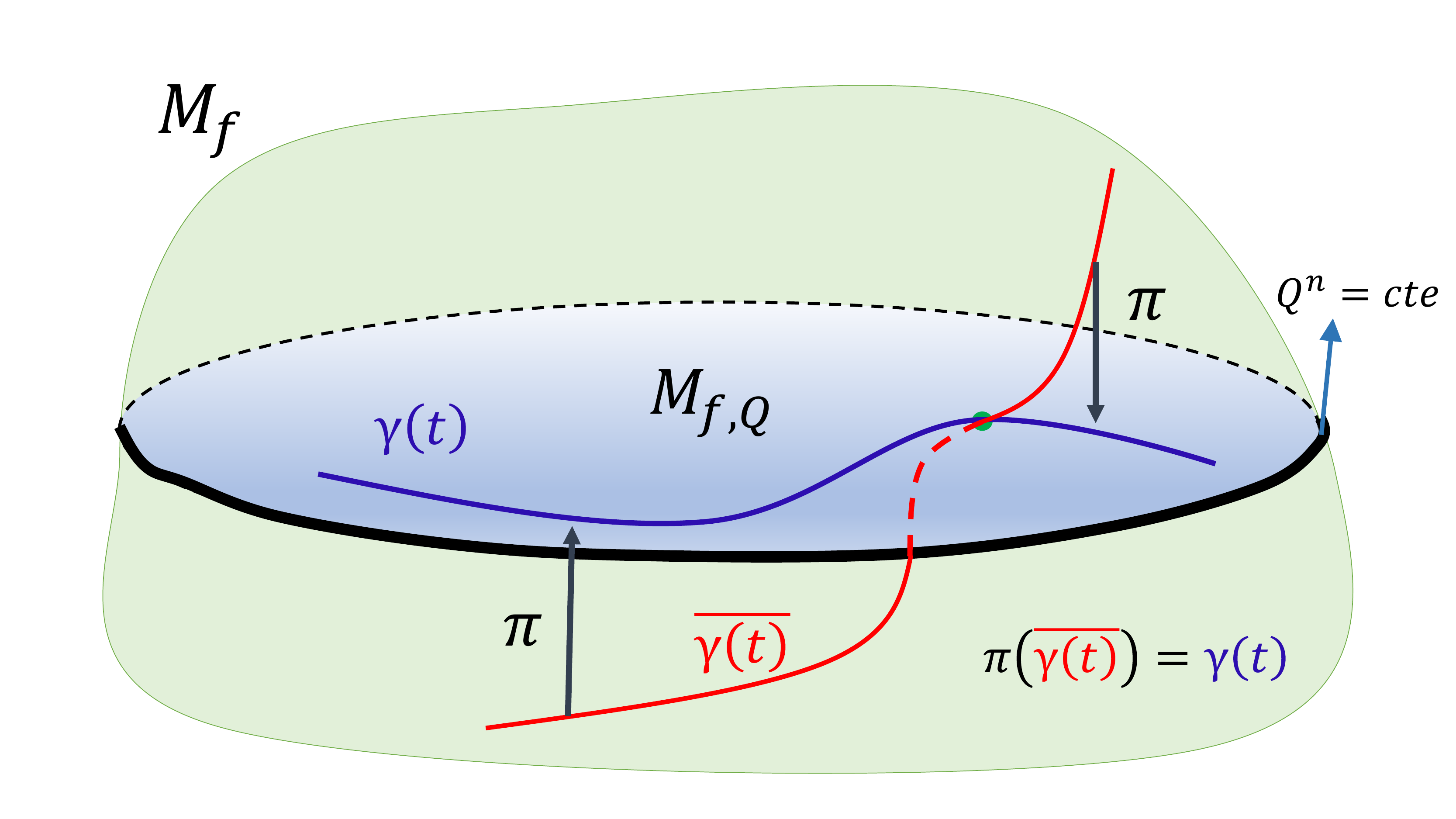}
    \caption{The dynamics defined by the reduced Hamiltonian system $(M_{f,Q},\omega\big|_{M_{f,Q}},K\big|_{M_{f}})$ with $n-1$ degrees of freedom is the projection of the dynamics of $X_{H}\big|_{M_{f}}$ into the submanifold $M_{f,Q}$. }
\label{FMfQ}    
\end{figure}

We conclude that solutions of the original Hamilton's equations of motion are also solutions of a reduced system of Hamilton's equations of motion, but the dynamics of the reduced system is not contained into the dynamics of the bigger one in general, see Figure \ref{FMfQ}. 

Only in some special cases, the existence of a particular integral allows us to construct a reduced Hamiltonian system whose dynamics is fully contained into the dynamics of the original one. Indeed, let us suppose that $f$ is a particular integral of $(M,\omega,H)$ and that we have a canonical transformation $(q,p)\mapsto (Q,P)$ for which $P_{n}=f$. In addition, let us assume that for the new Hamiltonian function $K$ the relation
\begin{equation}
\label{eqcuadratic}
\frac{\partial K}{\partial f}\ = \ b\,f \ ,
\end{equation}
holds for some function $b=b(Q,P)$ with real values on $f=0$, it implies $\dot{Q}^{n}= b\,f$. Thus, by putting $f=0$ the derivative $\dot{Q}^{n}$ vanishes and the equations (\ref{eqK}) in the new coordinates $(Q,P)$ reduces to the two equations $P_{n}=f=0$, $Q^{n}={\rm constant}$ and the system of $2(n-1)$ differential equations (\ref{reducedhamsys}). In this case, system (\ref{reducedhamsys}) is the system of Hamilton's equations of motion in the canonical coordinates $(Q^{1},\cdots,Q^{n-1},P_{1},\cdots,P_{n-1})$ for the reduced Hamiltonian system $(M_{f,Q},\omega\big|_{M_{f,Q}},K\big|_{M_{f,Q}})$ whose dynamics is defined by the restriction of the vector field $X_{H}$ to the submanifold $M_{f,Q}$, i.e., the dynamics of $(M_{f,Q},\omega\big|_{M_{f,Q}},K\big|_{M_{f,Q}})$ is contained in the dynamics of $(M,\omega,H)$. Condition (\ref{eqcuadratic}) is particularly fulfilled in natural physical systems where the new coordinates $(Q,P)$ are orthogonal, i.e., systems where the Lagrangian function is equal to the difference between the kinetic energy (a diagonal quadratic form in the generalized velocities) and the potential energy \cite{Arnold89} where the Hamiltonian function in the new coordinates $(Q,P)$ has the form:
\begin{equation}
K(Q,P)\ = \ \frac{1}{2}g^{jj}(Q)\,P_{j}^{2} \ + \ V(Q) \ .
\end{equation}
The reduction process described above is intimately related with a phenomena beyond separation of variables. Also, it is clear that each global constant of motion is a particular integral, but the converse statement is not, in general, true.  

Now, the same reduction process can be applied to the reduced system as well. For this purpose we require, in addition to $f$, a second particular integral $g$. Concretely, we can see that a function $g\in C^{\infty}(M)$ such that $f$ and $g$ obey, $\lbrace f,g\rbrace=a_{1}\,g+a_{2}\,f$ and $\lbrace g,H\rbrace=a_{3}\,g+a_{4}\,f$, being $a_i\in C^{\infty}(M)$ functions with real values on $f=g=0$, defines a particular integral of the reduced Hamiltonian system $K\big|_{M_{f}}$. 

In \cite{BCM2012} the idea of imposing two or more side conditions to the Hamilton-Jacobi equation (\ref{HJeq}) was presented previously. For two side conditions $S_{1}=S_{2}=0$, the so called consistency requirement is 
\begin{equation}
\label{eqtwoside}
\lbrace S_{i},S_{j}\rbrace=a^{ij}_{1}S_{1}+a^{ij}_{2}S_{2} \ ,
\end{equation}
($i,j=1,2,\ldots,n$). The above relations can be considered as the consistency conditions that guarantee the $S_i$ are integrals of motion for the Hamiltonian function $H=S_{3}$  of the system with $n$ degrees of freedom.

\subsection*{Particular integrals in the central-force problem}

To show how Theorem 1 can be employed in practice we first treat an elementary example in detail. More complicated examples will be presented later in Sections 3 and 4. Let us consider the classical central-force problem in $\mathbb{R}^3$. In Cartesian coordinates $\mathbf{r}=(x,y,z)=(q^1,q^2,q^3)$, the corresponding Hamiltonian reads
\begin{equation}
\label{hc}
    {\rm H}(q,p) \ = \ \frac{1}{2\,m}\,\big(\,p_x^2 \ + \ p_y^2 \ + \ p_z^2 \,\big)  \ + \ V \ ,
\end{equation}
here $m$ denotes the mass of the particle, $V=V(r)$ is a rotationally-invariant potential with $r= \sqrt{x^2+y^2+z^2}$ being the Euclidean distance to the center of the force, and $(p_x, p_y, p_z)=(p_1,p_2,p_3)$ are the canonical momentum variables, respectively. Making a canonical transformation to the non-orthogonal coordinate system $(r,\,\phi,\,\rho)=(Q^1,Q^2,Q^3)$
\[
r= \sqrt{x^2+y^2+z^2} \quad , \qquad  \phi= \tan^{-1}\bigg(\frac{y}{x}\bigg) \quad , \qquad \rho= \sqrt{x^2+y^2} \ ,
\]
the Hamiltonian (\ref{hc}) becomes
\begin{equation}
\label{hc2}
  K \ \equiv \ {\rm H}(Q,P) \ = \ \frac{1}{2\,m}\,\bigg(\,p_{\rho}^2  \ + \ p_{r}^2  \ + \ \frac{2\,\rho}{r}\, p_{\rho}\, p_{r} \ + \ \frac{1}{\rho^2}\, p_{\phi}^2\,\bigg) \ + \ V(r) \ ,
\end{equation}
with $p_r$ the canonical conjugate momentum to $r$, and so on. From (\ref{hc2}) we obtain the Hamilton's equations of motion:
\begin{equation}
\begin{aligned}
\label{eqmcp}
& {\dot r}  \ = \ \frac{1}{m}\,p_r \ + \ \frac{1}{m}\frac{\rho}{r}\,p_\rho 
\\ &
 {\dot p}_r \ = \ \frac{1}{m}\frac{\rho}{r^2}\, p_{\rho}\, p_{r} \ - \ V^\prime(r)
  \\ & 
 {\dot \rho}  \ = \ \frac{1}{m}\,p_\rho \ + \ \frac{1}{m}\frac{\rho}{r}\,p_r
 \\ &
 {\dot p}_\rho \ = \ -\frac{1}{m\,r}\, p_{\rho}\, p_{r} \ + \ \frac{1}{m\,\rho^3}\,p_\phi^2 
 \\ & 
 {\dot \phi}  \ = \ \frac{1}{m\,\rho^2}\,p_\phi
 \\ & 
 {\dot p}_\phi  \ = \ 0 \ ,
\end{aligned}    
\end{equation}
$V^\prime(r)\equiv \frac{d}{dr}V$.
In this case, $p_{\phi}$ is a global integral of motion ($\lbrace p_{\phi},\,K\rbrace=0$, $\phi$ is a cyclic coordinate), thus, a trivial particular integral. On the other hand, $p_\rho$ is not a global constant of motion for $K$. The Poisson bracket
\[
\lbrace p_{\rho},\,K\rbrace \ = \ - \bigg(\frac{1}{m\,r}\,p_r\bigg)\,p_\rho \ + \ \bigg(\frac{1}{m\,\rho^3}\,p_{\phi}\,\bigg)\,p_{\phi} \ ,
\]
in general, does not vanish. However, on the invariant manifold defined by $p_{\phi}=0$, the variable $p_{\rho}$ becomes a particular integral for the reduced Hamiltonian 
\[
G \  \equiv \  K\big|_{M_{p_\phi}} \  = \  \frac{1}{2\,m}\,\big(\,p_{\rho}^2  \ + \ p_{r}^2 \  + \  \frac{2\,\rho}{r}\, p_{\rho}\, p_{r}   \,\big) \  + \  V(r) \ ,
\]
$\lbrace p_{\rho},\,G\rbrace  =  - (\,\frac{1}{m\,r}\,p_r\,)\,p_\rho$. Thus, at $p_{\phi}=0$ the system (\ref{eqmcp}) reduces to ${\dot p}_\phi=0$, $\phi = {\rm constant}$ and the equations:
\begin{equation}
\begin{aligned}
\label{eqmcpred}
& {\dot r}  \ = \ \frac{1}{m}\,p_r \ + \ \frac{1}{m}\frac{\rho}{r}\,p_\rho 
\\ &
 {\dot p}_r \ = \ \frac{1}{m}\frac{\rho}{r^2}\, p_{\rho}\, p_{r} \ - \ V^\prime(r)
  \\ & 
 {\dot \rho}  \ = \ \frac{1}{m}\,p_\rho \ + \ \frac{1}{m}\frac{\rho}{r}\,p_r
 \\ &
 {\dot p}_\rho \ = \ -\frac{1}{m\,r}\, p_{\rho}\, p_{r} \ ,
\end{aligned}    
\end{equation}
which are the Hamilton's equations of the reduced Hamiltonian $G$. Moreover, at $p_{\phi}=0$ and $p_{\rho}=0$, the above system (\ref{eqmcpred}) further reduces to ${\dot p}_\rho=0$ and the equations 
\begin{equation}
\begin{aligned}
\label{eqmcpred2}
& {\dot r}  \ = \ \frac{1}{m}\,p_r 
\\ &
 {\dot p}_r \ = \  -  V^\prime(r)
  \\ & 
 {\dot \rho}  \ = \  \frac{1}{m}\frac{\rho}{r}\,p_r\ .
\end{aligned}    
\end{equation}
Here, the equations (\ref{eqmcpred2}) define an autonomous dynamical system of third order but they do not correspond to a system of Hamilton's equations of motion. Eventually, using the integrals $p_\phi$ and $p_\rho$ we arrive to the one-dimensional reduced Hamiltonian
\begin{equation}
    J \  \equiv \  G\big|_{M_{p_\rho}} \ = \ \frac{1}{2\,m}\, p_{r}^2  \ + \ V(r) \ .
\end{equation}
It describes all the trajectories of the original three-dimensional Hamiltonian system (\ref{hc}) possessing zero total angular momentum.

In the previous example, the geometric structure on the phase space (the symplectic structure) is present in the defining equation of the Poisson bracket (\ref{PB}) whilst in (\ref{CPB}) the Poisson bracket is written in a selected set of canonical coordinates.

In analogy to global integrals, if for a given Hamiltonian system there is a  sufficient number of particular integrals then we could reduce the system of the Hamilton's equations of motion up to the point that we can solve it by quadratures. Again, these solutions may not be the most generic ones. We address this reduction in the next Section.

\section{Particular integrability}

For Hamiltonian systems, the notion of complete integrability in the Liouville sense is equivalent to the statement that the corresponding solutions of the Hamilton's equations can be obtained by quadratures, provided that a sufficient number of functionally-independent integrals of motion (first integrals) in involution are known. In the previous Section, we have defined the concept of particular integral. Now, we are interested in the notion of particular complete Liouville integrability. To this end, below we introduce the property of particular involution. 

Explicitly,  the set $C^{\infty}(M)$ has structure of a ring, namely $(+,.)$, and it has structure of a module over itself as well; if we consider $C^{\infty}(M)$ as a module over itself then for given $f_{1},\cdots,f_{k}\in C^{\infty}(M)$ the set $\lbrace a_{1}f_{1}+\cdots+a_{k}f_{k} \ : \ a_{1},\cdots,a_{k}\in C^{\infty}(M)\rbrace$ is a submodule of $C^{\infty}(M)$ called the submodule finitely generated by $f_{1},\cdots,f_{k}$ \cite{Rotman2002}. The set $C^{\infty}(M)$ is an algebra with the Poisson bracket $\lbrace,\rbrace$, if we have a collection of functions $f_{1},\cdots,f_{k}\in C^{\infty}(M)$ we can look at the subalgebra generated by the Poisson brackets $\lbrace f_{i},f_{j}\rbrace$. The Bour-Liouville theorem on integrability by quadratures \cite{AKN2006,Prykarpatsky99,Azuaje2022} states that the solutions of the Hamilton's equations can be obtained by quadratures, provided the existence of a sufficient number of functionally-independent integrals of motion that form a solvable Lie algebra with the Poisson bracket. For this result the structure of a module of $C^{\infty}(M)$ over itself is not considered, i.e., only linear combinations of the form $\ c^{ij}_{1}f_{1}\,+\,\cdots\,+\,c^{ij}_{n}f_{k}$, with $c^{ij}_{1},\cdots,c^{ij}_{k}\in\mathbb{R}$, are involved. It is worth remarking that in this paper the linear combinations are considered on $C^{\infty}(M)$ as a module over itself, so the coefficients are smooth functions on $M$. Thus, a linear combination with constants coefficients (like in the Bour-Liouville theorem) is a special case of the present consideration.
\begin{de}
We say that the functions  $f_{1},\cdots,f_{k}\in C^{\infty}(M)$ are in particular involution if 
\begin{equation}
\label{eqPS}
\lbrace f_{i},f_{j}\rbrace\ = \ a^{ij}_{1}f_{1}\,+\,\cdots\,+\,a^{ij}_{n}f_{k} \ ,
\end{equation}
for some functions $a^{ij}_{1},\cdots,a^{ij}_{k}\in C^{\infty}(M)$ with real values on $f_{1}=\cdots=f_{k}=0$.
\end{de}
These functions $a^{ij}_{1},\cdots,a^{ij}_{k}\in C^{\infty}(M)$ in the previous definition are such that \[ \lim_{x\to x_{0}} a^{ij}_{s}(x)\ = \ a^{ij}_{s}(x_{0})\in\mathbb{R}\quad , \qquad \forall \,\,x_{0}\in M_{f}=\lbrace x\in M: \ f_{1}(x)=\cdots=f_{k}(x)=0\rbrace \ . \] Next, if the functions  $f_{1},\cdots,f_{k}\in C^{\infty}(M)$ are in particular involution then every element of the Lie subalgebra generated by the Poisson brackets $\lbrace f_{i},f_{j}\rbrace$ is an element of the submodule finitely generated
by $f_{1},\cdots,f_{k}$. We can see the similarity between condition (\ref{eqtwoside}) and condition (\ref{eqPS}); in fact, the second one follows the pattern of the first one. We can also see that the condition (\ref{eqPS}) is more general than the standard involution condition, i.e., a collection of function $f_{1},\cdots,f_{k}\in C^{\infty}(M)$ in involution trivially satisfies condition (\ref{eqPS}). The following theorem generalizes the Liouville theorem on integrability by quadratures.
\begin{te}
\label{tePI}
If $f_{1},\cdots,f_{n}\in C^{\infty}(M)$ are functionally-independent and in particular involution functions, then for the Hamiltonian system defined on $(M,\omega)$ with Hamiltonian function $f_{1}$, we can find special solutions (not the most general ones) of the Hamilton's equations of motion by quadratures. 
\end{te}
To be more rigorous, the essence of theorem \ref{tePI} can  be stated as follows: If $f_{1},\cdots,f_{n}\in C^{\infty}(M)$ are functionally-independent and in particular involution functions, then for each Hamiltonian system $(M,\omega,f_{i})$ with $n$ degrees of freedom, the solutions of the Hamilton's equations of motion that live in $M_{f}=\lbrace x\in M: \ f_{1}(x)=\cdots=f_{n}(x)=0\rbrace$ can be found by quadratures. Before giving a proof of theorem \ref{tePI}, we present a brief review of Lie integrability by quadratures of vector differential equations in $\mathbb{R}^{n}$. Lie’s theorem on integrability establishes that if for a vector field we have sufficient linearly independent symmetries that generate a solvable Lie algebra over the Lie bracket of vector fields, then the solutions of the equations of motion of the dynamical system defined by such vector field can be found by quadratures. The concept of solvable Lie algebra is more general than the one of Abelian Lie algebra, indeed, every Abelian Lie algebra is trivially a solvable Lie algebra \cite{Gilmor2008,Fecko2006}. Let $v$ be a smooth vector field on $\mathbb{R}^{n}$, the Lie's theorem \cite{Lie, Kozlov,AKN2006,CFGR2015,Azuaje2022,Cetal2016} indicates sufficient conditions for the dynamical system defined by the vector differential equation $\dot{x}=v(x)$ to be integrable by quadratures, it reads as follows:

\begin{te} {\rm \textbf{[S. Lie]}}
\label{teLie}
Let $u_{1},\ldots u_{n}$ be linearly independent smooth vector fields on $\mathbb{R}^{n}$. If $u_{1},\ldots,u_{n}$ are symmetries of the vector field $v$ and they generate a solvable Lie algebra with the Lie bracket $[,]$ of vector fields, then the solutions of the vector differential equation $\dot{x}=v(x)$ can be found by quadratures.
\end{te}

Now, we are ready to present the proof of theorem \ref{tePI}.
\begin{proof}
Let $f_{1},\cdots,f_{n}$ be functionally-independent and in particular involution functions in $C^{\infty}(M)$, i.e.
\begin{equation}
\lbrace f_{i},f_{j}\rbrace=a^{ij}_{1}f_{1}+\cdots+a^{ij}_{n}f_{n} \ ,
\end{equation}
for some functions $a^{ij}_{1},\cdots,a^{ij}_{n}\in C^{\infty}(M)$ with real values on $f_{1}=\cdots=f_{k}=0$. Since the functions $f_{1},\cdots,f_{n}$ are functionally independent on $M$ then at any point $p$ in an open dense subset of $M$ the differential maps $df_{1}|_{p},\ldots,df_{n}|_{p}$ are linearly independent; this means that $p$ is a regular point of the differentiable function $F:N\longrightarrow \mathbb{R}^{n}$ defined by $F=(f_{1},\ldots,f_{n})$, so from the Regular Level Set Theorem \cite{Lee2012} we have that $M_{f}=\lbrace x\in M: \ f_{1}(x)=\cdots=f_{n}(x)=0\rbrace$ is a smooth submanifold of $M$ of dimension dim$(M)-n=2n-n=n$.

On the other hand, the vector fields $X_{f_{1}},\ldots,X_{f_{n}}$ are tangent to the submanifold $M_{f}$ because $X_{f_{i}}f_{j}=\lbrace f_{j},f_{i}\rbrace$ which is zero on $M_{f}$. Without lost of generality, let us take the Hamiltonian function $H=f_{1}$. So the solutions of the Hamilton's equations that live in $M_{f}$ are integral curves of $X_{H}|_{M_{f}}$, i.e., the solutions of the Hamilton's equations that live in $M_{f}$ are solutions of the vector differential equation $\dot{x}=X_{H}(x)$ with $x=(x^{1},\ldots,x^{n})$ being local coordinates on $M_{f}$. We know that the assignment $f\mapsto X_{f}$ is a Lie algebra antihomorphism between the Lie algebras $(C^{\infty}(M),\lbrace,\rbrace)$ and $(\mathfrak{X}(M),[,])$ \cite{Lee2012}, i.e., $X_{\lbrace f,g\rbrace}=-[X_{f},X_{g}]$, so we have that the vector fields $X_{f_{1}}|_{M_{f}},\cdots,X_{f_{n}}|_{M_{f}}$ are symmetries of $X_{H}|_{M_{f}}$ and generate an Abelian (therefore solvable) Lie algebra with the Lie bracket of vector fields. Therefore, by Theorem \ref{teLie} we have that the solutions of the vector differential equation $\dot{x}=X_{H}(x)$ can be found by quadratures. In other words, the solutions of the Hamilton's equations of motion that live in $M_{f}$ can be found by quadratures.
\end{proof}

It is worth remarking that theorem \ref{tePI} is valid for arbitrary Hamiltonian systems, of course for normal Hamiltonian systems as well. In general, the reduced system defined by $X_{H}$ on $M_{f}$ is not a Hamiltonian system. For instance, observe that $M_{f}$ might be an odd dimensional manifold. Even though it may be even dimensional it does not necessarily inherit the symplectic structure on $M$. Also, the functions $f_{1},\cdots,f_{n}$ in theorem \ref{tePI} may be particular integrals or even global constants of motion. Indeed, the condition (\ref{eqPS}) is more general than the defining condition for a particular integral; the equation $\lbrace f_{i},H\rbrace=a_{i}f_{i}$ is a special case of equation $\lbrace f_{i},H\rbrace=a^{ij}_{1}f_{1}+\cdots+a^{ij}_{n}f_{n}$. If the functions $f_{1},\cdots,f_{n}$ are particular integrals satisfying condition (\ref{eqcuadratic}) (for example for a normal Hamiltonian system), then the equations of motion for the dynamical system defined by $X_{H}$ on $M_{f}$ are of the form $\dot{x}=constant$ so they are trivially solved by quadratures. For general Hamiltonian systems, natural or not, we propose the notion of particular integrability as follows: 
\begin{de}
A Hamiltonian system $(M,\omega,H)$ with $n$ degrees of freedom is said to be particularly integrable if there exist $n$ functionally-independent functions $H=f_{1},\cdots,f_{n}\in C^{\infty}(M)$ in particular involution.
\end{de}
For a Hamiltonian system with $n$ degrees of freedom the maximum number of functionally-independent global integrals in involution is $n$. The same result holds in the case of particular integrals, as stated in the following lemma 
\begin{prop}
For a given Hamiltonian system we have that the maximum number of independent functions in particular involution is equal to the number of degrees of freedom.
\end{prop}
\begin{proof}
Let us suppose that for a Hamiltonian system with $n$ degrees of freedom there exist more than $n$ functionally-independent functions $f_{1},\cdots,f_{s}$ with $s>n$ that satisfy condition (\ref{eqPS}). Since the functions $f_{1},\cdots,f_{s}$ are functionally-independent on $M$ then at any point $p$ in an open dense subset of $M$ the differential maps $df_{1}|_{p},\ldots,df_{n}|_{p}$ are linearly independent. This means that $p$ is a regular point of the differentiable function $F:N\longrightarrow \mathbb{R}^{s}$ defined by $F=(f_{1},\ldots,f_{s})$. Now, from the Regular Level Set Theorem \cite{Lee2012} we have that $M_{f}=\lbrace x\in M: \ f_{1}(x)=\cdots=f_{s}(x)=0\rbrace$ is a smooth submanifold of $M$ of dimension dim$(M)-s=2n-s<s$. The vector fields $X_{f_{1}},\ldots,X_{f_{s}}$ are tangent to the submanifold $M_{f}$ because $X_{f_{i}}f_{j}=\lbrace f_{j},f_{i}\rbrace$ vanishes on $M_{f}$. Also, the vector fields $X_{f_{1}}|_{M_{f}},\ldots,X_{f_{s}}|_{M_{f}}$ on $M_{f}$ are linearly independent at any point $p$ in an open dense subset of $M$, so they generate a vector subspace of $T_{p}M_{f}$ of dimension $s$ but the dimension of $T_{p}M_{f}$ is $2n-s$ which is less that $s$, hence, we have a contradiction. Therefore, we conclude that the maximum number of functionally-independent functions in particular involution is $n$ (the number of degrees of freedom).
\end{proof}

\subsection*{Particular integrability in the two-body Coulomb problem in a constant magnetic field}

An integrable system in the sense of Liouville (when the number of degrees of freedom equals the amount of independent constants of motion in involution \cite{AKN2006,BBT2003,Das}) is particularly integrable; but the converse statement is not necessarily true. For example, let us consider in $\mathbb{R}^2$ the classical analog of the planar hydrogen atom in a perpendicular constant magnetic field. This system consists in two Coulomb charges of opposite sign moving on a plane with an attractive interaction $-\frac{1}{r}$. Additionally, they are subjected to the presence of an external constant magnetic field.  
First, we introduce collective center-of-mass and relative variables in the corresponding Hamiltonian. The so called total Pseudomomentum ${\bbf{K}}$ \cite{Gorkov,Avr78} is a global conserved quantity
associated with the center-of-mass motion which
can be used to rewrite the Hamiltonian in a transparent and physically simple form
\begin{equation}
\label{Hmf}
{\cal H} \ = \ {\cal H}_{\rm CM} \ + \  {\cal H}_{\rm coupling} \ + \ {\cal H}_{\rm rel}    \ ,
\end{equation}
where
\begin{equation}
\begin{aligned}
 & {\cal H}_{\rm CM} \ = \  \frac{1}{2\,M}\,\bbf{K}^2 
\\ &
{\cal H}_{\rm coupling} = \frac{e}{M}\left({\bbf{B}}\times {\bbf{K}}\right).{\bbf{r}}
\\ &
{\cal H}_{\rm rel} = \frac{1}{2\,\mu}\left({\bbf{p}}-\frac{e_{\rm eff}}{2}{\bbf{B}}\times
{\bbf{r}}\right)^2 \ + \ \frac{e^2\,B^2\,r^2}{2\,M} 
\ - \ \frac{e^2}{r} \ ,
\end{aligned}
\end{equation}
here $\mu=\frac{m_1\,m_2}{M}$ and $M=m_1+m_2$ are the reduced and total mass of the system, respectively, $(-e)$ is the electron charge, $m_1$ the electron mass, $e_{\rm eff}=e\,\frac{m_1-m_2}{M}$ is an effective charge and ${\bbf{B}}$ is the constant
magnetic field vector which is assumed to be perpendicular to the plane (along the $z$-axis). The variables
$({\bbf{R}},{\bbf{K}})$ and $({\bbf{r}},{\bbf{p}})$ are the canonical pairs
for the CM and internal motion, respectively. The corresponding phase space is eight-dimensional. 

This fundamental physical system (\ref{Hmf}) possesses three global integrals of motion: the total Hamiltonian $\cal H$ and the two components of the Pseudomomentum ${\bbf{K}}$. Nevertheless, the system is known to be chaotic; a fourth global integral does not exist. However, in the special case ${\bbf{K}}=\bbf{0}$ (the system at rest) the $z-$component $\ell_z= ({\bbf{r}}\times
{\bbf{p}})_z$ of the relative angular momentum turns out to be a conserved quantity (i.e., only at ${\bbf{K}}=\bbf{0}$ the Poisson bracket $\lbrace \ell_z,\,{\cal H}\rbrace$ vanishes) and the system becomes particularly integrable \cite{MA2013}. The corresponding quantities in particular involution are given by $({\cal H},\,{\bbf{K}}=\bbf{0},\,\ell_z)$.

\section{$N-$body system: particular integrals and symmetry reduction }

In general, in the Hamiltonian of an $N-$body system in $\mathbb{R}^d$ the scalar potential only depends on a certain number of effective variables $n_0 < n$, where $n= N\,d$ is the number of degrees of freedom. For instance, in the case of symmetrically-invariant potentials (assuming $d > N-2$) $n_0$ is nothing but the $\frac{N\,(N-1)}{2}$ relative distances between bodies. If we are interested in the trajectories (solutions of the Hamilton's equations of motion) that solely depend on these relative distances, certain particular integrals can be used to construct a reduced Hamiltonian which precisely governs their dynamics \cite{AM2021, AM2022, Miller2018}.

Concretely, let us we consider $N$ particles moving in a $d-$dimensional Euclidean space $\mathbb{R}^d$. The Hamiltonian of the system takes the form
\begin{equation}
\label{HN}
    {\cal H} \ = \ {\cal T} \ + \ V \ = \ \sum_{i=1}^N \frac{\mathbf{p}_i^2}{2\,m_i} \ + \ V(|\mathbf{r}_i-\mathbf{r}_j|) \ ,
\end{equation}
where $\mathbf{r}_i$, $i=1,2,\ldots,N$, is the vector position of particle $i$, $m_i$ its mass and $\mathbf{p}_i$ is the associated canonical momentum. The dimension of the phase space $\Gamma$ is dim$\,\Gamma=2\,N\,d$. It is assumed that the potential $V=V(r_{ij})$ solely depends on the relative distances between particles, $r_{ij}\equiv |\mathbf{r}_i-\mathbf{r}_j| $. 

Since $V=V(r_{ij})$, the total linear momentum 
\[
\mathbf{P} \ = \ \mathbf{p}_1 \ + \ \mathbf{p}_2 \ + \ldots \ \mathbf{p}_N \ ,
\]
of the center of mass vectorial variable
\[
\mathbf{R} \ = \ \frac{m_1\,\mathbf{r}_1+m_2\,\mathbf{r}_2+\ldots+ m_N\,\mathbf{r}_N}{m_1+m_2+\ldots+m_N} \ ,
\]
is a first integral of motion (in the Liouville sense), its Poisson bracket with (\ref{HN}) is identically zero $\{ \mathbf{P},\,  {\cal H} \}=0$. The dynamics of the variables ($\mathbf{R},\,\mathbf{P}$) corresponds to that of a free particle, and it is fully decoupled from the dynamics of the internal relative motion of the system. Thus, the center-of-mass motion can be separated out completely, and the internal relative motion of the system is characterized by $d(N-1)$ degrees of freedom. The dimension of the corresponding (reduced) phase space
\begin{equation}
\label{}
\Gamma_{\rm rel} \ \equiv \ \Gamma(\mathbf{R}={\mathbf{0}},\,\mathbf{P}={\mathbf{0}}) \ , 
\end{equation}
is dim$\,\Gamma_{\rm rel}=2\,d\,(N-1)$. In this space of relative motion, let be
\begin{equation}
\{q_1,\,q_2,\,q_3,\ldots,\,q_{d(N-1)}\} \ ,
\end{equation}
any set of generalized coordinates. In particular, one possibility is to take the $(N-1)$ vectorial Jacobi coordinates defined by the linear relations
\begin{equation}
\label{Jacobi}
\begin{aligned}
    & {\bf r}^{(J)}_{1} \ = \ \sqrt{\frac{m_1\,m_2}{m_1+m_2}}\left({\bf r}_{2}\ - \ {\bf r}_{1}\right)
\\ &
{\bf r}^{(J)}_{2} \ = \ \sqrt{\frac{(m_1+m_2)\,m_3}{m_1+m_2+m_3}}\left({\bf r}_{3}\ - \ \frac{m_1\,{\bf r}_{1}+m_2\,{\bf r}_{2}}{m_1+m_2}\right)
\\ &
\vdots
\\ &
{\bf r}^{(J)}_{(N-1)} \ = \ \sqrt{\frac{(m_1+m_2+\ldots\,m_{N-1})\,m_N}{m_1+m_2+\ldots\,m_N}}\left({\bf r}_{N}\ - \ \frac{m_1\,{\bf r}_{1}+m_2\,{\bf r}_{2}+ \ldots+m_N\,{\bf r}_{N-1}}{m_1+m_2+\ldots\,m_{N-1}}\right)
     \  .
\end{aligned}
\end{equation}
In these coordinates, the $d(N-1)$-dimensional kinetic energy term ${\cal T}_{\rm rel}$ of the relative motion becomes diagonal. Explicitly,
\begin{equation}
\label{Tflat-diag}
   2\,{\cal T}\ = \ {\cal T}_{\rm center-of-mass}^{(d)} \ + \ {\cal T}_{\rm rel}^{(d(N-1))} \ = \ {{\bf P}}_{{}_{\mathbf R}}^{2} \ + \ {({\bf p}^{(J)}_{1})}^2 \ + \ {({\bf p}^{(J)}_{2})}^2\ + \ldots + \ {({\bf p}^{(J)}_{N-1})}^2 \ ,
\end{equation}
here ${\bf p}^{(J)}_{i}$ is the $d-$dimensional canonical momemtum associated with the vector ${\bf r}^{(J)}_{i}$. The kinetic energy of relative motion is the sum of kinetic energies in the Jacobi coordinate directions. Therefore, we arrive to the reduced Hamiltonian 
\begin{equation}
\label{Hred}
    {\cal H}_{\rm rel} \ = \ {\cal T}_{\rm rel}^{(d(N-1))} \ + \ V \ .
\end{equation}
Now, assuming $d > N-2$, in the space of relative motion parametrized by Jacobi-vectorial variables we introduce a further change of coordinates
\begin{equation}
\label{CC4b}
\{{\bf r}^{(J)}_{1},\,{\bf r}^{(J)}_{2},\,\ldots,\,{\bf r}^{(J)}_{N-1}\}\ \Rightarrow \
\{\, \{ r_{ij}\},\,\{\phi_k\}\,  \} \ ,
\end{equation}
where $r_{ij}$ are the $\frac{N(N-1)}{2}$ relative distances between particles whereas the collection \{$\phi_k$\}, $k=1,2,3,\ldots,\frac{1}{2} (N-1) (2 d-N)$, stands for any set of variables such that the transformation (\ref{CC4b}) is canonical.  For the potential $V$ does not involve the $\phi-$variables, it follows that the hypersurface
\begin{equation}
\label{IC4b}
\Gamma_r \ \equiv \ \Gamma_{\rm rel}(\, p_{_{\phi_1}} =\,p_{_{\phi_2}}=\,\ldots\,= \,p_{_{\phi_{\frac{1}{2} (N-1) (2 d-N)}}}=0\,) \quad ; \, \quad  \Gamma_r \ \subset \ \Gamma_{\rm rel} \subset \Gamma \ , 
\end{equation}
is an invariant manifold. In other words, any trajectory of ${\cal H}_{\rm rel}$ (\ref{Hred}) with initial conditions $\ p_{_{\phi_1}} =\,p_{_{\phi_2}}=\,p_{_{\phi_3}}=\,\ldots\,= \,p_{_{\phi_{\frac{1}{2} (N-1) (2 d-N)}}}=0$ will remain in $\Gamma_r$ under time evolution. It should be noted that the $\phi-$variables are mostly not cyclic coordinates. In general, they will appear in the kinetic energy term ${\cal T}_{\rm rel}^{(d(N-1))}$ explicitly. Hence, unlike the center-of-mass momentum $\mathbf{P}$, the associated momenta $p_{\phi}$ are not first Liouville integrals but \textit{particular integrals} of motion. They are conserved quantities for specific initial conditions only. The trajectories possessing zero total angular momentum necessarily lie in $\Gamma_r$.

\subsection*{The $\rho-$representation}

Explicitly, in the coordinates $\rho_{ij}=r^2_{ij}$, the reduced Hamiltonian (\ref{Hred}) reads
\begin{equation}
\label{HRN1}
  {\cal H}_{\rm rel}^{(d(N-1))}=\  2\,\sum_{i<j}^N \bigg(\frac{m_i+m_j}{m_i m_j} \bigg) \rho_{ij}\, p^2_{\rho_{ij}}\ +\ \sum_{i \neq j,i\neq k, j< k}^N\,\frac{2}{m_i}\,(\rho_{ij} + \rho_{ik} - \rho_{jk})\,p_{\rho_{ij}}\,p_{\rho_{ik}}\ +\ \Omega \ + \ V  \ ,
\end{equation}
where every term in $\Omega=\Omega(p_{\rho},\,p_{\phi},\,\rho,\,\phi)$ is either a product of 2 \textit{angular} momenta $p_{\phi}$ or a product of a radial momentum $p_{\rho}$ and an angular momentum \cite{Miller2018}.
The corresponding Hamilton-Jacobi equation is
\begin{equation}
\label{HJa}
 {\cal H}_{\rm rel} ^{(d(N-1))}(p_{\rho},\,p_{\phi},\,\rho,\,\phi)\ =\ E\ ,
\end{equation}
where
\[
   p_{\rho_{ij}}=\frac{\pa W}{\pa \rho_{ij}} \qquad ; \quad  p_{\phi_k}\
   =\ \frac{\partial W}{\pa{\phi_k}}\ ,
\]
for the Hamilton's characteristic function $W=W(\rho_{ij},\phi_k)$. Now, using the particular integrals $p_\phi$ one can construct, on the invariant manifold
$\Gamma_r$ (\ref{IC4b}), a further reduced Hamiltonian ${\cal H}_{\rho}$ given by \cite{Miller2018}

\begin{equation}
\label{HRNrho}
  {\cal H}_{\rho}\ = \  2\,\sum_{i<j}^N \bigg(\frac{m_i+m_j}{m_i m_j} \bigg) \rho_{ij}\, p^2_{\rho_{ij}}\ +\ \sum_{i \neq j,i\neq k, j< k}^N\,\frac{2}{m_i}\,(\rho_{ij} + \rho_{ik} - \rho_{jk})\,p_{\rho_{ij}}\,p_{\rho_{ik}}\ + \ V(\rho_{ij})  \ ,
\end{equation}
which can be considered as describing a  classical top with variable tensor of inertia
in an external potential \cite{Arnold89}. The form of (\ref{HRNrho}) is independent of the set of variables \{$\phi_k$\} used in (\ref{HRN1}). For the Hamiltonian ${\cal H}_{\rho}$, the number of degrees of freedom is equal to the number of relative distances between particles, namely $\frac{N(N-1)}{2}$. They are nothing but the variables on which the original potential $V$ depends. It governs all the trajectories of the $N-$body system possessing zero total angular momentum. Therefore, assuming $d > N-2$, in the transformations ${\cal H} \rightarrow {\cal H}_{\rm red} \rightarrow {\cal H}_{\rho} $ the number of degrees of freedom is reduced from $2\,d\,N$ to $\frac{N(N-1)}{2}$. In the Appendix \ref{rhorep}, we present the explicit form of ${\cal H}_{\rho}$ for the lowest cases $N=2,3,4$.

\subsection*{The volume-variable representation}

Finally, we describe one more reduction elaborated in \cite{AM2021, AM2022, Miller2018}. To this end, we first introduce the notion of the $(N-1)-$dimensional non-degenerate polytope of interaction. Explicitly, it is the polytope whose $N$-vertices are identified with the individual vector positions of the $N$ particles. In this case, the edges of the polytope coincide with the relative distances $r_{ij}$ between particles. The goal is to use the geometric properties of this object to define new generalized coordinates that we call \textit{volume-variables}. These variables together with suitable particular integrals will be used to construct, from ${\cal H}_{\rho}$ (\ref{HRNrho}), another reduced Hamiltonian ${\cal H}_{\rho} \rightarrow {\cal H}_{\rm vol}$.

Now, the polytope of interaction possesses elements of different dimensionality such as vertices, edges, faces, cells. Concretely, an $(N-1)$-dimensional polytope is bounded by a number of ($N-2$)-dimensional facets, which are themselves polytopes whose facets are ($N-3$)-dimensional ridges of the original polytope and so on. In particular, the content of a polytope, as well as of any other of its lower-dimensional facets, is a generalized volume, namely a \textit{hypervolume}. 

The first volume variable ${\cal V}_{N-1}$ is the content (squared) of the polytope. It is completely defined by the relative distances $r_{ij}$ and can be computed using the Cayley-Menger determinant. At fixed $N$, the dimension of ${\cal V}_N$ is given by $[{\rm distance}]^{2(N-1)}$. The second volume variable ${\cal V}_{N-2}$, is the sum of all the contents (squared) of their $(N-2)-$dimensional facets. Clearly, dim(${\cal V}_{N-2}$)$=[{\rm distance}]^{2(N-2)}$. Similarly, we define the variable ${\cal V}_{N-k}$, here $k=3,4,\ldots,(N-1)$. In total, the number of volume variables is equal to $(N-1)$, and all of them can be expressed in terms of the variables $r_{ij}$ using linear combinations of Cayley-Menger determinants. In addition to the original assumption $V=V(r_{ij})$, in this representation we consider potentials $V=V({\cal V}_{N-k})$ which solely depends on the volume variables. Therefore, assuming $d > N-2$, in the transformations ${\cal H} \rightarrow {\cal H}_{\rm red} \rightarrow {\cal H}_{\rho}\rightarrow {\cal H}_{\rm vol} $ the number of degrees of freedom is effectively reduced from $2\,d\,N$ to $N-1$.

In the Appendix \ref{volrep}, for the lowest cases $N=2,3,4,5,6$ we display the form of the reduced Hamiltonian ${\cal H}_{\rm vol}$.

\section{Conclusions}

In this study, we have presented the notion of a particular integral from a geometric point of view. For natural Hamiltonian systems, this concept plays an important role in the partial reduction of the Hamilton's equations. Especially, within the formalism of symplectic geometry, the concept of complete particular integrability was treated in detail using an adaptive application of the classical results due to Lie and Liouville. Any complete integrable system is  particularly integrable but even a non-linear non-integrable system may exhibit certain regular regions where the property of particular integrability occurs.
To illustrate the key aspects of the present symplectic theory approach, physically significant systems were studied in detail.
For future work there are interesting open questions. For instance, the notion of particular integrability in the cases of cosymplectic, contact, and cocontact geometries. An approach to developing a theory of particular integrability in the quantum case (where a natural geometric framework is not evident) remains unsolved as well. It will be done elsewhere.

\section{Acknowledgments}

The author R. Azuaje wishes to thank CONACYT (Mexico) for the financial support through a postdoctoral fellowship under the program Estancias Posdoctorales por México 2022. In addition, the author A. M. Escobar Ruiz wishes to thank Prof. Turbiner for useful discussions and helpful comments.

\appendix

\section{Symmetry reduction of the $N-$body system using particular integrals}

\subsection{Reduced Hamiltonian ${\cal H}_\rho$}
\label{rhorep}

\subsubsection*{Case $N=2$}

In the simplest two-body case $N=2$, for any $d$, the reduced Hamiltonian ${\cal H}_{\rho}$ (\ref{HRNrho}) in the space of relative motion takes the form 

\begin{equation}
\label{HRNrhoN2}
  {\cal H}_{\rho}\ = \  2\,\bigg(\frac{m_1+m_2}{m_1\, m_2} \bigg) \rho_{12}\, p^2_{\rho_{12}}\  + \ V(\rho_{12})  \ ,
\end{equation}
where $\rho_{12}=r_{12}^2$ is the relative distance (squared) between particles. It correspond to the well-known reduced Hamiltonian for the two-body central-force problem at zero angular momentum. In this case, ${\cal H}_{\rho}$ (\ref{HRNrhoN2}) plays the role of a particular integral of motion. It is conserved on the invariant manifold of zero angular momentum only.

\subsubsection*{Case $N=3$}

In the three-body case $N=3$, for any $d>1$, the reduced Hamiltonian ${\cal H}_{\rho}$ (\ref{HRNrho}) reads \cite{AM2021}

\begin{equation}
\label{HRNrhoN3}
\begin{aligned}
 & {\cal H}_{\rho}\ = \  2\,\bigg(\frac{m_1+m_2}{m_1\, m_2} \bigg) \rho_{12}\, p^2_{\rho_{12}}\  + \ 2\,\bigg(\frac{m_1+m_3}{m_1\, m_3} \bigg) \rho_{13}\, p^2_{\rho_{13}}\  + \ 2\,\bigg(\frac{m_2+m_3}{m_2\, m_3} \bigg) \rho_{23}\, p^2_{\rho_{23}}
 \\ &
 \  + \ \frac{2}{m_1}\,(\rho_{12} + \rho_{13} - \rho_{23})\,p_{\rho_{12}}\,p_{\rho_{13}}\  + \ \frac{2}{m_2}\,(\rho_{12} + \rho_{23} - \rho_{13})\,p_{\rho_{12}}\,p_{\rho_{23}}
\\ &
 \  + \ \frac{2}{m_3}\,(\rho_{13} + \rho_{23} - \rho_{12})\,p_{\rho_{13}}\,p_{\rho_{23}}
\ + \ V(\rho_{12},\,\rho_{13},\,\rho_{23})  \ .
\end{aligned}
\end{equation}
In this case, from the original $4d-$dimensional problem, $d>1$, we arrive to the above 3-dimensional Hamiltonian, solutions of which are also solutions of the original one \cite{AM2021}. 

\subsubsection*{Case $N=4$}

In a similar manner, for the 4-body system we immediately arrive at the reduced Hamiltonian \cite{AM2022}

\begin{equation}
\label{HredrhoN4}
\begin{aligned}
  {H}_{\rho} \  & = \ 2\,\bigg[\,\frac{\rho_{12}\,p_{\rho_{12}}^2}{m_{12}} \ + \ \frac{\rho_{13}\,p_{\rho_{13}}^2}{m_{13}} \ + \  \frac{\rho_{14}\,p_{\rho_{14}}^2}{m_{14}} \ + \   \frac{\rho_{23}\,p_{\rho_{23}}^2}{m_{23}} \ + \  \frac{\rho_{24}\,p_{\rho_{24}}^2}{m_{24}} \ + \ \frac{\rho_{34}\,p_{\rho_{34}}^2}{m_{34}} \ + \ \frac{\rho_{12}+\rho_{13}-\rho_{23}}{m_1}\,p_{\rho_{12}}\,p_{\rho_{13}}
\\ &
\qquad   \ +
  \ \frac{\rho_{12}+\rho_{23}-\rho_{13}}{m_2}\,p_{\rho_{12}}\,p_{\rho_{23}}\ + \
\frac{\rho_{23}+\rho_{13}-\rho_{12}}{m_3}\,p_{\rho_{23}}\,p_{\rho_{13}}
\ + \
\frac{\rho_{24}+\rho_{14}-\rho_{12}}{m_4}\,p_{\rho_{24}}\,p_{\rho_{14}}
\\ &
\qquad   \ + \
   \frac{\rho_{12}+\rho_{14}-\rho_{24}}{m_1}\,p_{\rho_{12}}\,p_{\rho_{14}} \ + \
   \frac{\rho_{12}+\rho_{24}-\rho_{14}}{m_2}\,p_{\rho_{12}}\,p_{\rho_{24}}\ + \
   \frac{\rho_{13}+\rho_{34}-\rho_{14}}{m_3}\,p_{\rho_{13}}\,p_{\rho_{34}}
   \\ &
\qquad   \ + \
   \frac{\rho_{14}+\rho_{34}-\rho_{13}}{m_4}\,p_{\rho_{14}}\,p_{\rho_{34}}  \ + \
   \frac{\rho_{13}+\rho_{14}-\rho_{34}}{m_1}\,p_{\rho_{13}}\,p_{\rho_{14}} \ + \
   \frac{\rho_{23}+\rho_{24}-\rho_{34}}{m_2}\,p_{\rho_{23}}\,p_{\rho_{24}}
   \\ &
\qquad   \ + \
   \frac{\rho_{23}+\rho_{34}-\rho_{24}}{m_3}\,p_{\rho_{23}}\,p_{\rho_{34}} \ + \
   \frac{\rho_{24}+\rho_{34}-\rho_{23}}{m_4}\,p_{\rho_{24}}\,p_{\rho_{34}}   \,\bigg] \ + \  V(\rho_{ij})     \ ,
\end{aligned}
\end{equation}
with $m_{ij}\equiv \frac{m_i\,m_j}{m_i+m_j}$. Again, (\ref{HredrhoN4}) plays the role of a particular integral of motion.

\subsection{Reduced Hamiltonian ${\cal H}_{\rm vol}$}
\label{volrep}

For simplicity, we assume that all $N-$particles have the same mass $m_i=1$, $i=1,2,\ldots,N$.

\subsubsection*{Case $N=2$}

In the case $N=2$, there is only one volume variable, namely ${\cal V}_{1}=\rho_{12}=r_{12}^2$. Therefore, at $N=2$ the two Hamiltonians ${\cal H}_{\rho}$ (see (\ref{HRNrhoN2})) and ${\cal H}_{\rm vol}$ coincide.

\subsubsection*{Case $N=3$}

At $N=3$, there are only two volume variables given by
\begin{equation}
\begin{aligned}
& {\cal V}_{1} \ = \ \rho_{12} \ + \ \rho_{13} \ + \ \rho_{23}
\\ &
{\cal V}_{2} \ = \ \frac{1}{16}\big(\,  2\,\rho_{12}\,\rho_{23} \ +\  2\,\rho_{12}\,\rho_{13} \ +\  2\,\rho_{23}\,\rho_{13} \ -\  \rho_{12}^2 \ -\  \rho_{23}^2 \ -\  \rho_{13}^2 \,\big)
\end{aligned}
\end{equation}
respectively. In particular, ${\cal V}_{2}$ is the area (squared) of the triangle of interaction formed by taking the individual vector positions of the three particles as vertices. In these variables, assuming $V=V({\cal V}_{1},{\cal V}_{2})$, we obtain the reduced Hamiltonian \cite{AM2021}

\begin{equation}
\label{}
  {\cal H}_{\rm vol}\ = \ 6\,{\cal V}_{1}\,P_1^2 \ + \  \frac{1}{2}{\cal V}_{1}\,{\cal V}_{2}\,P^2_2 \ + \   24\,{\cal V}_{2}\,P_1\,P_2  \ + \ V({\cal V}_{1},{\cal V}_{2})\ ,
\end{equation}
here $P_k\equiv P_{{\cal V}_{k}}$, $k=1,2$, denote the corresponding canonical conjugate momentum variables.  

\subsubsection*{Case $N=4$}

In the next case $N=4$, there are three volume variables:

\begin{equation}
\label{VVV}
\begin{aligned}
{\cal V}_3 \ & = \ \frac{1}{144} \bigg[\, \rho_{23} \rho_{34} \rho_{12}+ \rho_{24} \rho_{34} \rho_{12}+\rho_{14} \rho_{23} \rho_{12}
+ \rho_{13} \rho_{24} \rho_{12} \ + \  \rho_{13} \rho_{34} \rho_{12}+ \rho_{14} \rho_{34} \rho_{12}
\\ &
 \ + \ \rho_{13} \,\rho_{14}\, \rho_{23}+ \rho_{13} \rho_{14} \rho_{24}+ \rho_{13} \rho_{23} \rho_{24}+ \rho_{14} \rho_{23} \rho_{24} + \rho_{14} \rho_{23} \rho_{34}+\rho_{13}\, \rho_{24} \,\rho_{34}
\\ &
\ - \  \rho_{14}\, \rho_{23}^2- \rho_{13} \,\rho_{24}^2- \rho_{14}^2 \,\rho_{23}-\rho_{34} \,\rho_{12}^2- \rho_{34}^2 \,\rho_{12}- \rho_{13}^2 \,\rho_{24}
\ - \ \rho_{13} \, \rho_{14} \,\rho_{34}
\\ &
\ - \  \rho_{23} \,\rho_{24}\, \rho_{34}\ - \ \rho_{13}\, \rho_{23} \,\rho_{12}\ - \  \rho_{14}\, \rho_{24}\, \rho_{12}\,\bigg] \ ,
\end{aligned}
\end{equation}

\begin{equation}
\label{}
\begin{aligned}
{\cal V}_2 \  = \ & \frac{1}{16} \,\bigg[ \big(2 \rho_{13} \,\rho_{12}+2 \rho_{23}\, \rho_{12}+2 \rho_{13}\, \rho_{23}-\rho_{13}^2-\rho_{23}^2-\rho_{12}^2  \big) \ + \
\\ &
\big( 2 \rho_{14}\, \rho_{12}+2\rho_{24}\, \rho_{12}+2 \rho_{14}\, \rho_{24}-\rho_{12}^2-\rho_{14}^2-\rho_{24}^2\big) \ + \
\\ &
\big(2 \rho_{14}\, \rho_{13}+2 \rho_{34}\, \rho_{13}+2 \rho_{14}\, \rho_{34}-\rho_{14}^2-\rho_{34}^2 -\rho_{13}^2\big) \ +\
\\ &
\big( 2 \rho_{24}\, \rho_{23}+2 \rho_{34}\, \rho_{23}+2 \rho_{24}\, \rho_{34}-\rho_{24}^2-\rho_{34}^2-\rho_{23}^2 \big)\,\bigg] \nonumber  \ ,
\end{aligned}
\end{equation}

\begin{equation}
\label{}
{\cal V}_1 \ = \ \rho_{12} \ + \ \rho_{13} \ + \ \rho_{14} \ + \ \rho_{23} \ + \ \rho_{24} \ + \ \rho_{34}   \nonumber  \ .
\end{equation}

In these variables, assuming $V=V({\cal V}_{1},{\cal V}_{2},,{\cal V}_{3})$, we arrive to the following reduced Hamiltonian \cite{AM2022}:

\begin{equation}
\begin{aligned}
\label{}
{\cal H}_{\rm vol} \ = \
 & \ 8\,{\cal V}_{1}\,P^2_{1} \ + \  \frac{1}{2}\, \big({\cal V}_{1}\,{\cal V}_{2}\,+\,108\,{{\cal V}_{3}}\big)\,P^2_{2}\ + \  \frac{2}{9}\,{\cal V}_{2}\,{\cal V}_{3}\,P^2_{3}
\\
 &  + \ 32\,{\cal V}_{2}\,P_1\,P_2\ + \ 48\,{{\cal V}_{3}}\,P_1\,P_3\ + \ 2\,{\cal V}_{1}\, {\cal V}_{3}\,P_2\,P_3   \ .
\end{aligned}
\end{equation}
 For the cases $N=2,3,4$ the reduced Hamiltonian ${\cal H}_{\rm vol}$ was known \cite{Miller2018}. Below, we present new results\footnote{They are in complete agreement with unpublished results obtained in the quantum case (private correspondence with A. V. Turbiner and W. Miller Jr.) .} for the cases $N=5$ and $N=6$.

\subsubsection*{Case $N=5$}

For the case $N=5$, there are four volume variables. The corresponding reduced Hamiltonian takes the form:

\[
{\cal H}_{\rm vol}  \, = \,10\,{\cal V}_{1}\,P^2_{1}\ + \ \frac{1}{2}\,\bigg({\cal V}_{1}\,{\cal V}_{2}+135\,{\cal V}_{3}\,\bigg)\,P^2_{2}
  \ + \ \frac{2}{9}\,({\cal V}_{2}\,{\cal V}_{3} \ + \ 12 {\cal V}_{1}\,{\cal V}_{4})\,P^2_{3}
  \ + \ \frac{1}{8}\,{\cal V}_{3}\,{\cal V}_{4}\,P^2_{4}
\]
\begin{equation}
\begin{aligned}
\label{}
& + \ 40\,{\cal V}_{2}\,P_1\,P_2\ + \ 60\,{\cal V}_{3}\,P_1\,P_3\ + \ 80\,{\cal V}_{4}\,P_1\,P_4
  \ + \ 2({\cal V}_{1}\,{\cal V}_{3}\,+\,160\,{\cal V}_{4})\,P_2\,P_3\ + \  3\,{\cal V}_{1}\,{\cal V}_{4}\,P_2\,P_4
  \ + \ \frac{8}{9}{\cal V}_{2}\,{\cal V}_{4}\,P_3\,P_4 \ .
\end{aligned}
\end{equation}

\subsubsection*{Case $N=6$}

Finally, for the case $N=6$, there are five volume variables. The 5-dimensional reduced Hamiltonian is given by:

\begin{equation}
\begin{aligned}
\label{}
{\cal H}_{\rm vol} & \ = \ 12\,{\cal V}_{1}\,P^2_{1}\ + \  \big(\frac{1}{2}\,{\cal V}_{1}\,{\cal V}_{2}\,+\,81\,{\cal V}_{3}\big)\,P^2_{2}\ + \     \big(\frac{2}{9}\,{\cal V}_{2}\,{\cal V}_{3}+\frac{8}{3}\,{\cal V}_{1}\,{\cal V}_{4}+\frac{2000}{3}\,{\cal V}_{5}\big)\,P^2_{3} \ + \
\\ &
\big(\frac{1}{8}\,{\cal V}_{3}\,{\cal V}_{4} + \frac{25}{24}\, {\cal V}_{2} {\cal V}_{5}\big)\,P^2_{4}\ + \  \frac{2}{25}\, {\cal V}_{4}\,{\cal V}_{5}\,P^2_{5}\ + \ 48\,{\cal V}_{2}\,P_1\,P_2\ + \ 72\,{\cal V}_{3}\,P_1\,P_3\ + \ 96\,{\cal V}_{4}\,P_1\,P_4\ + \
\\ &
 120\,{\cal V}_{5}\,P_1\,P_5\ + \ 2\,({\cal V}_{1}\,{\cal V}_{3}\,+\,192\,{\cal V}_{4})\,P_2\,P_3\ + \  (3\,{\cal V}_{1}\,{\cal V}_{4} + 750\, {\cal V}_{5})\,P_2\,P_4\ + \
 4\, {\cal V}_{1}\,{\cal V}_{5}\,P_2\,P_5\ + \
\\ &
\big(\frac{8}{9}{\cal V}_{2}\,{\cal V}_{4} + \frac{100}{9}\, {\cal V}_{1}\, {\cal V}_{5}\big)\,P_3\,P_4\ + \  \frac{4}{3}\, {\cal V}_{2}\,{\cal V}_{5}\,P_3\,P_5\ + \  \frac{1}{2}\,{\cal V}_{3}\,{\cal V}_{5}\,P_4\,P_5 \ .
\end{aligned}
\end{equation}

\end{document}